\documentclass[conference]{IEEEtran}
\usepackage{cite}

\usepackage{graphicx}
\usepackage[cmex10]{amsmath}
\usepackage{amsfonts}
\usepackage{algorithmic}

\usepackage{array}
\usepackage{amsthm}
\usepackage{amssymb}
\usepackage{subfigure}
\usepackage{mathtools}
\usepackage{epstopdf}
\usepackage{color}
\newtheorem{theorem}{Theorem}
\newtheorem{lemma}[theorem]{Lemma}

\newtheorem{proposition}{Proposition}

\usepackage{enumitem}

\begin{document}
\renewcommand{\thefootnote}{}	
%
\title{Self-Calibration for Massive MIMO with Channel Reciprocity and Channel Estimation Errors}
%
%
%

\author{\IEEEauthorblockN{De~Mi\IEEEauthorrefmark{1}, Lei~Zhang\IEEEauthorrefmark{2}, Mehrdad~Dianati \IEEEauthorrefmark{3}\IEEEauthorrefmark{1}, Sami~Muhaidat\IEEEauthorrefmark{4}\IEEEauthorrefmark{1},  Pei~Xiao\IEEEauthorrefmark{1} and Rahim~Tafazolli\IEEEauthorrefmark{1}}
	\IEEEauthorblockA{\IEEEauthorrefmark{1}Institute for Communication Systems, University of Surrey, United Kingdom}
	\IEEEauthorblockA{\IEEEauthorrefmark{2}School of Engineering, University of Glasgow, United Kingdom}
	\IEEEauthorblockA{\IEEEauthorrefmark{3}Warwick Manufacturing Group, University of Warwick, United Kingdom}
	\IEEEauthorblockA{\IEEEauthorrefmark{4}Department of Electrical and Computer Engineering, Khalifa University, United Arab Emirates}
	Email:d.mi@surrey.ac.uk, lei.zhang@glasgow.ac.uk, m.dianati@warwick.ac.uk, \\ muhaidat@ieee.org, p.xiao@surrey.ac.uk, r.tafazolli@surrey.ac.uk}

\maketitle

\begin{abstract}
	In time-division-duplexing (TDD) massive multiple-input multiple-output (MIMO) systems, channel reciprocity is exploited to overcome the overwhelming pilot training and the feedback overhead. However, in practical scenarios, the imperfections in channel reciprocity, mainly caused by radio-frequency mismatches among the antennas at the base station side, can significantly degrade the system performance and might become a performance limiting factor. In order to compensate for these imperfections, we present and investigate two new calibration schemes for TDD-based massive multi-user MIMO systems, namely, relative calibration and inverse calibration. In particular, the design of the proposed inverse calibration takes into account a compound effect of channel reciprocity error and channel estimation error. We further derive closed-form expressions for the ergodic sum rate, assuming maximum ratio transmissions with the compound effect of both errors. We demonstrate that the inverse calibration scheme outperforms the traditional relative calibration scheme. The proposed analytical results are also verified by simulated illustrations. 
\end{abstract}


{\scriptsize \footnote{We would like to acknowledge the support of the University of Surrey 5GIC (www.surrey.ac.uk/5gic) members for this work. This work was also supported in part by the European Commission under the 5GPPP project 5G-Xcast (H2020-ICT-2016-2 call, grant number 761498). The views expressed in this contribution are those of the authors and do not necessarily represent the project.}}

%
\IEEEpeerreviewmaketitle

\section{Introduction} \label{sec:intro}
Massive MIMO (multiple-input-multiple-output) is identified as a promising technological paradigm which has been proposed in order to meet some of 5G requirements, including data rates of 10-20 Gbps and a latency of less than 1 msec \cite{MMNG,MWC2015,3gpp22261f30,GZWCM2018}. Exploiting channel reciprocity, time-division-duplexing (TDD) operation enables the channel state information (CSI) acquisition in massive MIMO with an affordable overhead that is independent of the number of base station (BS) antennas \cite{T1}. Most prior studies assume the perfect channel reciprocity by constraining the time delay from the uplink (UL) to the downlink (DL) is within the coherence time of the channel \cite{HYTM2013JSAC, MMIMOULDL}. However, the assumption of the perfect reciprocity is unrealistic in practical systems even within the coherence time, due to the fact that radio-frequency (RF) transceivers introduce amplitude and phase mismatches between the UL and the DL \cite{BjornsonTIT2014}. The imperfect channel reciprocity contaminates the estimate of the effective channel response. This causes a significant degradation in the performance of linear precoding schemes, due to their sensitivity to the CSI accuracy. Our prior work \cite{ReciprocityErrorTCOM2016} has thoroughly investigated such performance degradation for two typical linear precoders, i.e., maximum radio transmission (MRT) and zero-forcing (ZF), with considerations of the imperfect channel estimation. The results in \cite{ReciprocityErrorTCOM2016} show that both MRT and ZF are severely affected by the compound effect of the reciprocity and estimation errors. Therefore, it is of great interest to study suitable reciprocity calibration schemes in the presence of both errors.

In principle, calibration schemes for a precoded TDD MIMO system contain the estimation of transmit (Tx) and receive (Rx) RF frontends' responses or equivalently calibration coefficients and the design of the calibration matrix to calibrate the precoders. For the calibration coefficients estimation, we explicitly focus on the so-called self-calibration in massive MIMO systems since it can be implemented at the BS side only and without exchanging calibration pilots between the BS and user terminals (UTs) \cite{ChnReConf2RE}. Prior studies proposed different methods to realise self-calibration in massive MIMO systems \cite{MutualCouplingTWC2016,ArgosConfShort}. In \cite{MutualCouplingTWC2016}, the use of antenna coupling at the BS has been proposed to measure the calibration coefficients. This method is very sensitive to the scatterings near the BS antennas. A practical study in \cite{ArgosConfShort} presented a method where an additional RF transceiver is used as a reference to exchange calibration pilots with other BS antennas' transceivers. However, the method in \cite{ArgosConfShort} is sensitive to the placement of the reference transceiver. In order to obtain reliable estimates of the calibration coefficients, additional calibration circuits can be applied at the BS, as one example shown in \cite{RefNewCali2003}. The self-calibration scheme in \cite{RefNewCali2003,ArgosConfShort} is known as relative calibration, which has been widely considered in the context of the massive MIMO system \cite{ScalSynReTWC2014,WenceTCOMM15}. 

Although a number of recent progress has been developed from the perspective of the reciprocity calibration, various issues still remain open, such that, to the best knowledge of the authors, there have been no calibration schemes taking into account the compound effect of the reciprocity and channel estimation errors. In order to fill these research gaps among the existing literature, we consider a low-cost calibration circuit that presented in \cite{RefNewCali2003}, and expand it into the TDD massive MU-MIMO system. Such design enables the BS to estimate the RF responses reliably, also to perform the relative calibration or what we call ``inverse calibration" (in the sense that it is based on the inverse of the calibration coefficients). More importantly, with considerations of the compound effect of the reciprocity and estimation errors, we provide an in-depth analysis of the performance of a precoded massive MIMO system with the aforementioned two calibration algorithms. The proposed analytical results are verified via Monte-Carlo simulations. The rest of the paper is organised as follows. In Section~\ref{sec:sys}, we describe the TDD massive MIMO system model with imperfect channel estimation and the reciprocity error model due to the RF mismatch. In Section~\ref{sec:new}, the description of the relative and inverse calibration algorithms comes after a discussion of the considered calibration circuit. The performance evaluation of these two algorithms is given in Section~\ref{sec:performance}. Simulation results and conclusions are provided in Section~\ref{sec:re} and Section~\ref{sec:con}, respectively. Some of the detailed derivations are given in the appendices.

\section{System Model}\label{sec:sys}
A TDD massive MU-MIMO system is considered in this paper as illustrated in Fig.~\ref{fig:fig1}. It consists of $M$ antennas at the BS, each antenna is connected with an individual RF chain. 
In addition, $K$ single-antenna UTs ($M \gg K$) are served in the same time and frequency resources. We assume that the time delay from the UL channel estimation to the DL transmission is less than the coherence time of the channel, ensuring that the propagation channels on the UL and DL are equal. 

\begin{figure}[!t]
	\centering
	\includegraphics[angle=0,width=0.45\textwidth]{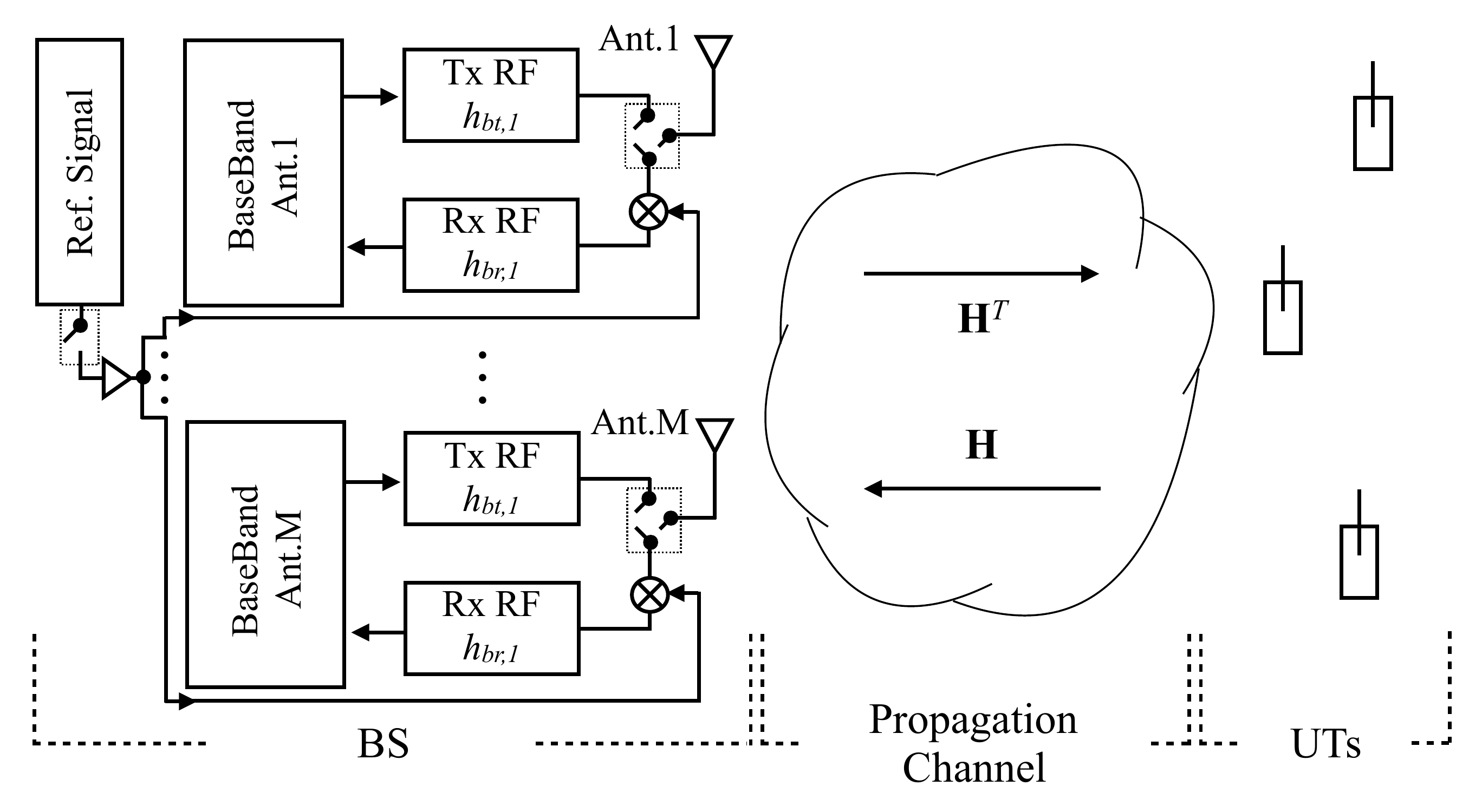} 
	\caption{A TDD massive MU-MIMO System with calibration circuits.}
	\label{fig:fig1}
\end{figure}

As shown in Fig.~\ref{fig:fig1}, we denote the UL and DL propagation channels by $ \mathbf{H} \in \mathbb{C}^{M \times K} $ and $ \mathbf{H}^{T} $ respectively, whose entries follow independent identically distributed (i.i.d.) $C\mathcal{N}(0, 1)$. We consider the same model of the BS RF frontends response as that set up in \cite{ReciprocityErrorTCOM2016}. Briefly speaking, $ M \times M $ diagonal matrices $ \mathbf{H}_{br} $ and $ \mathbf{H}_{bt} $ represent the effective response matrices of the Rx and Tx RF frontends at the BS, whose $ i $-th diagonal entries are $ h_{br,i} = A_{br,i}\textnormal{exp}(j\varphi_{br,i}) $ and $ h_{bt,i} = A_{bt,i}\textnormal{exp}(j\varphi_{bt,i}) $, respectively, where $A$ denotes amplitude, $\varphi$ denotes phase. We model both amplitude and phase reciprocity errors as independent truncated Gaussian random variables \cite{ReciprocityErrorTCOM2016}. Without loss of generality, the statistical magnitudes of these truncated Gaussian distributed variables are assumed to be static within the considered coherence time of the channel or even longer period, e.g., minutes \cite{ArgosConfShort}.

We then consider the UL training protocol based on the minimum mean-square error (MMSE) channel estimation as in \cite{HYTM2013JSAC}, by taking into account the effect of $\mathbf{H}_{br}$ and $\mathbf{H}_{bt}$. More specifically, in TDD massive MIMO systems, UTs first transmit the orthogonal UL pilots of length $\tau_u$ to BS, where $\tau_u \geq K$. Therefore, the MMSE estimate of the actual UL channel response $ \mathbf{H}_{u} $ can be given by \cite{HYTM2013JSAC}
\begin{equation} \label{eq:HuE}
\hat{\mathbf{H}}_{u} = a\mathbf{H}_{br}\mathbf{H} + b\mathbf{N}_{u},
\end{equation}
where the estimation-error-related parameters are given by
\begin{equation} \label{eq:esterr}
a = \dfrac{\tau_u\rho_u}{\tau_u\rho_u+1}, \: b = \dfrac{\sqrt{\tau_u\rho_u}}{\tau_u\rho_u+1}.
\end{equation}
In addition, the $ M \times K $ noise matrix $ \mathbf{N}_{u} $ is the channel estimation noise matrix with i.i.d. $ C\mathcal{N}(0, 1) $ elements and is independent of $\hat{\mathbf{H}}_{d}$, and $\rho_u$ denotes the expected UL transmit SNR. Then the BS uses the transpose of $\hat{\mathbf{H}}_{u}$ as the estimate of the DL channel $\hat{\mathbf{H}}_{d}$, i.e., $\hat{\mathbf{H}}_{d} = \hat{\mathbf{H}}_{u}^{T}$, whereas the actual effective DL channel is $\mathbf{H}_{d} = \mathbf{H}^{T}\mathbf{H}_{bt}$. By comparing $\hat{\mathbf{H}}_{d}$ and $\mathbf{H}_{d}$, we can rewrite the DL channel estimate $\hat{\mathbf{H}}_{d}$ as 
\begin{equation} \label{eq:HdECompare}
\hat{\mathbf{H}}_{d} = a\mathbf{H}_{d}\mathbf{E} + b\mathbf{N}_{u}^{T},
\end{equation} 
where $\mathbf{E} = \mathbf{H}_{bt}^{-1}\mathbf{H}_{br}$ denotes the channel reciprocity error. We can see from \eqref{eq:HdECompare} a compound effect of an \textit{additive} distortion, $\mathbf{N}_{u}$, caused by the imperfect channel estimation, and a \textit{multiplicative} distortion, $\mathbf{E}$, caused by the imperfect channel reciprocity. 
Let a $K \times 1$ vector $ \mathbf{s} = [s_1, \cdots s_k, \cdots, s_K]^{T} $ denote the symbol to be transmitted to $ K $ UTs, where the normalised symbol power per user is assumed, i.e., $\mathbb{E}\left\lbrace \lvert s_k \rvert^2 \right\rbrace = 1$, for $ k = 1,2,\cdots,K $. We also assume that the symbols of different users are independent. The BS applies an $M \times K$ linear precoding matrix $\mathbf{W}$ to map the symbol vector $\mathbf{s}$ into an $M \times 1$ transmit signal vector to the BS antennas. We use $\mathbf{x}$ to denote this transmit signal vector, which is given by
\begin{equation} \label{eq:x}
\mathbf{x} = \sqrt{\rho_{d}}\lambda\mathbf{W}\mathbf{s},
\end{equation}
where $ \rho_{d} $ denotes the average transmit power at the BS (note that the expression \eqref{eq:x} implies that the power is equally allocated to each UT in this work), and $\lambda$ is a normalisation parameter to satisfy the transmission power constraint at the BS such that
\begin{equation}\label{eq:powercon}
\mathbb{E}\left\lbrace \lVert \mathbf{x} \rVert^{2} \right\rbrace = \mathbb{E}\left\lbrace \lVert \sqrt{\rho_{d}}\lambda\mathbf{W}\mathbf{s} \rVert^{2} \right\rbrace = \rho_{d}.
\end{equation}
Hence, $\lambda$ can be calculated as follows:
\begin{equation} \label{eq:lamda}
\lambda = \sqrt{\dfrac{1}{\mathbb{E}\left\lbrace \textnormal{tr}\left( \mathbf{W}\mathbf{W}^{H} \right)  \right\rbrace }}.
\end{equation}
 
Based on \eqref{eq:x}, the received signals of all $K$ UTs can be expressed in a vector form as 
\begin{equation} \label{eq:yd}
\mathbf{y} = \mathbf{H}_{d}\mathbf{x} + \mathbf{n} = \sqrt{\rho_{d}}\lambda\mathbf{H}^{T}\mathbf{H}_{bt}\mathbf{W}\mathbf{s} + \mathbf{n},
\end{equation}
where the $K \times 1$ vector $ \mathbf{n} $ denotes the DL received noise for all $K$ UTs, whose $ k^{\text{th}} $ element $ n_{k} \sim C\mathcal{N}(0, \sigma_k^2) $. We assume that $\sigma_k^2 = 1, \forall k$. Therefore, $ \rho_{d} $ can also be treated as the average input SNR for the DL transmission. For the $k^{\text{th}}$ UT, we have
\begin{equation} \label{eq:yk}
y_{k}  =  \sqrt{\rho_{d}}\lambda\mathbf{h}_{k}^{T}\mathbf{H}_{bt}\mathbf{w}_{k}s_{k} + \!\! \sum_{i = 1, i \neq k}^{K} \!\!\sqrt{\rho_{d}}\lambda\mathbf{h}_{k}^{T}\mathbf{H}_{bt}\mathbf{w}_{i} s_{i} + n_{k},
\end{equation} 
where the $M \times 1$ vectors $\mathbf{h}_{k}$ and $\mathbf{w}_{k}$ are the $k^{\text{th}}$ column of $\mathbf{H}$ and $\mathbf{W}$ respectively. Note that the received signal $ y_{k} $ in \eqref{eq:yk} is decomposed into three terms. The first two terms accounts for the desired signal for the $ k^{\text{th}} $ UT and the inter-user interference from other $ K-1 $ UTs, respectively, while the last term represents the receiver noise. Since $\mathbf{W}$ is a function of the DL channel estimate instead of the actual DL channel response, the performance of linear precoding schemes is affected by the imperfect channel reciprocity. To address this issue, we introduce the calibration scheme in the following section.

\section{Self-Calibration} \label{sec:new}
To compensate for the imperfection of channel reciprocity, an $M \times M$ pre-precoding calibration matrix $\mathbf{B}$ can be introduced to compensate for the non-reciprocity \cite{WenceTCOMM15}, such that $\hat{\mathbf{H}}_{d}$ in \eqref{eq:HdECompare} becomes
\begin{equation} \label{eq:HdC}
\hat{\mathbf{H}}_{d,CL} = a\mathbf{H}_{d}\mathbf{E}\mathbf{B} + b\mathbf{N}_{u}^{T}\mathbf{B},
\end{equation}
where $\hat{\mathbf{H}}_{d,CL}$ represents the estimate of the DL channel response after applying calibration, which could be used to calculate the DL precoding matrix $\mathbf{W}$. The majority of the reported results on the reciprocity calibration has proposed the design concept of the calibration matrix $\mathbf{B}$ without considering the effect of channel estimation error \cite{ArgosConfShort,ScalSynReTWC2014,WenceTCOMM15}, e.g., in the case with $a\approx1$ and $b\approx0$. In such case, the minimum requirement to calibrate the BS antennas is that $\mathbf{EB} = c\mathbf{I}_{M} $, where the scalar $c \in \mathbb{C}_{\neq0}$ is multiplied by all calibration factors. The scalar $c$ is arbitrary. $\mathbb{C}_{\neq0}$ denotes the set of non-zero complex numbers. Thus, this does not change the direction of the precoding beamformer \cite{schenk2008rf}. 

The acquisition of the calibration matrix $\mathbf{B}$ contains two steps: 1) the estimation of $\mathbf{H}_{bt}$, $\mathbf{H}_{br}$, and 2) finding $\mathbf{B}$ based on the estimates of $\mathbf{H}_{bt}$ and $\mathbf{H}_{br}$. In the following, we first present a calibration circuit used for the first step. We then discuss two calibration algorithms obtained based on the considered circuit design, and their relationship with the channel estimation error.

\subsection{RF Frontend Response Measurement}
Motivated by the calibration method for conventional MIMO systems in \cite{RefNewCali2003}, we present a circuit design as shown in Fig.~\ref{fig:fig1}, to achieve the measurement of the effective response matrix of the BS RF frontend for the massive MIMO system. Particularly, switching units attached to each antenna have three modes: ``Tx/Rx" mode (the antenna connects to Tx or Rx RF frontend), ``Link" mode (Tx and Rx RF frontends are connected) and ``Null" mode (no connection); A reference signal source is split and equally injected at each Rx Rf frontend by couplers. Then the measurement of $\mathbf{H}_{bt}$ and $\mathbf{H}_{br}$ can be carried out, which contains two steps which we call ``self connection" and ``half connection". 

During the self connection, all switching units that attached to the BS antennas are set to be in ``Link" mode. The reference signal source is disconnected. The individual baseband of each BS antenna, taking $i$-th antenna element as an example, estimates the product of $h_{bt,i}$ and $h_{br,i}$ by sending a known signal $p_{i}$ simultaneously, such that $
r_{i} = h_{br,i}h_{bt,i}p_{i} + u_{i} $
received at the baseband of the $i$-th antenna, where the thermal noise $u_{i}$ has a negligible value due to the fact that the calibration SNR is usually sufficiently high, e.g., 20 dB in \cite{RefNewCali2003}. Thus the estimate of $\mathbf{H}_{bt}\mathbf{H}_{br}$, denoted by $\mathbf{R}_{\text{self}}$, is given by 
\begin{equation} \label{eq:Rs}
\mathbf{R}_{\text{self}} = \text{diag}\left( r_1/p_1, \cdots, r_i/p_i, \cdots, r_K/p_K \right). 
\end{equation}

During the half connection, all switching units that attached to the BS antennas are set to be in ``Null" mode. The reference signal source is equally injected at all Rx RF frontends. Let an $M \times 1$ vector $\mathbf{p}_{\text{ref}}$ be the reference signal vector with duplicate entries. Then the collective received signal vector at all the BS antenna baseband is given by $
\mathbf{r}_{h} = \mathbf{H}_{br}(\mathbf{p}_{\text{ref}} + \mathbf{u}_{h})$ ,
where $\mathbf{r}_{h}$ contains the received signals at each baseband that are sampled at the same time \cite{RefNewCali2003}. Again, the effect of the measurement noise $\mathbf{u}_{h}$ is assumed to be trivial in this work due to the assumption on the high calibration SNR. Hence, the estimate of $\mathbf{H}_{br}$, denoted by $\mathbf{R}_{\text{half}}$, is given by
\begin{equation} \label{eq:Rh}
\mathbf{R}_{\text{half}} = \text{diag}(\mathbf{r}_{h})\left( \text{diag}(\mathbf{p}_{\text{ref}})\right)^{-1}.
\end{equation}

Note that, as discussed in Section~\ref{sec:sys}, the reciprocity-error-related parameters, or equivalently $\mathbf{H}_{bt}$ and $\mathbf{H}_{br}$, are relatively static, i.e., they change in a much slower rate compared to the variations of the channel state. Hence, once the measurement of the BS RF responses, i.e., $\mathbf{R}_{\text{self}}$ and $\mathbf{R}_{\text{half}}$, is reliably obtained, it can be applied within the the coherence time of the channel or even longer period \cite{ArgosConfShort}. 

\subsection{Design of the Calibration Matrix}
Based on the measurement in \eqref{eq:Rs} and \eqref{eq:Rh}, the calibration matrix $\mathbf{B}$ can be calculated. The study in \cite{RefNewCali2003} considers the relative calibration scheme, where the calibration matrix, denoted by $\mathbf{B}_{\text{RC}}$, is given by $
\mathbf{B}_{\text{RC}} = \mathbf{R}_{\text{self}}\left(\mathbf{R}_{\text{half}}^{2}\right)^{-1}. $
As discussed at the beginning of this section, the widely-used relative calibration ignores the effect of the imperfect channel estimation, which can result in the estimation error amplification, and additionally cause the enhancement of the inter-user interference. More specifically, it can be seen from \eqref{eq:HdC} that the use of the calibration matrix $\mathbf{B}_{\text{RC}}$ may amplify the power of the estimation noise (equivalently, channel estimation error), which can even outweigh the benefit of calibration in certain cases, such as in the low region of $\rho_u$. This motivates us to design a calibration matrix without amplifying the estimation error.

In this work, we present a calibration scheme to compensate for the effect of the channel reciprocity error, as well as to reduce the noise power of the UL channel estimation, or equivalently reduce the estimation noise variance. To this end, we consider a calibration matrix, denoted by $\mathbf{B}_{\text{IC}}$, which is given by
\begin{equation}
\mathbf{B}_{\text{IC}} = \mathbf{R}_{\text{half}}^{\ast}\left( \mathbf{R}_{\text{half}}\mathbf{R}_{\text{self}}^{\ast}\right)^{-1}.
\end{equation}
In the ideal scenario, e.g., in the sufficiently high calibration SNR regime, the calibration matrix $\mathbf{B}_{\text{IC}}$ is equivalent to the inverse of the product of $\mathbf{H}_{bt}^{\ast}$ and $\mathbf{H}_{br}$. Thus, we name this calibration scheme as ``Inverse Calibration". Such a scheme can ensure that noise power of UL channel estimation after calibration is equal to or even less than that before calibration, since the expected value of the product of $\mathbf{H}_{bt}^{\ast}$ and $\mathbf{H}_{br}$ is greater than or equal to one \cite{ReciprocityErrorTCOM2016}. We shall evaluate the performance of the inverse calibration and the traditional relative calibration in the following section.

\section{Performance Evaluation} \label{sec:performance}
In this section, we analyse the ergodic sum rate to evaluate the performance of the inverse calibration and the widely-used relative calibration. 
We consider the simplest precoder, i.e., maximum ratio transmission. Note that our theoretical analysis contends with the compound effects on the system performance of the additive channel estimation error and multiplicative channel reciprocity error.
 
Recall that in \eqref{eq:yk}, we denote the desired signal power of the $k^{\text{th}}$ UT by $P_s$, and the its inter-user interference by $P_I$, where $P_s$ and $P_I$ are given by
\begin{equation} \label{eq:Ps}
P_{s} = \lvert\sqrt{\rho_{d}}\lambda\mathbf{h}_{k}^{T}\mathbf{H}_{bt}\mathbf{w}_{k}s_{k}\rvert^{2},
\end{equation}
\begin{equation} \label{eq:Pi}
P_{I} = \left|\sum_{i = 1, i \neq k}^{K}\sqrt{\rho_{d}}\lambda\mathbf{h}_{k}^{T}\mathbf{H}_{bt}\mathbf{w}_{i}s_{i}\right|^{2},
\end{equation}
respectively. Thus, the ergodic rate for the $k^{\text{th}}$ UT, denoted by $R_k$, can be given by
\begin{equation} \label{eq:Rk}
R_k = \mathbb{E}\left\lbrace \text{log}_2 \left( 1 + \dfrac{P_{s}}{P_{I} + \sigma_{k}^{2}}\right)  \right\rbrace .
\end{equation}
Let $R_K$ denote the ergodic sum rate of all $K$ UTs. Considering the approximation derived in \cite{ReciprocityErrorTCOM2016,LeiSINR2}, $R_K$ is given by
\begin{equation} \label{eq:RK2}
R_K = KR_k \approx K\text{log}_2 \left( 1 + {\mathbb E}\left\lbrace P_{s}\right\rbrace {\mathbb E}\left\lbrace\dfrac{1}{P_{I} + \sigma_{k}^{2}} \right\rbrace \right),
\end{equation} 
where
\begin{equation} \label{eq:EPi}
\mathbb{E}\left\lbrace \frac{1}{P_{I} + \sigma_{k}^{2}} \right\rbrace  = \dfrac{1}{\mathbb{E} \{P_{I} + \sigma_{k}^{2}\}} + \mathcal{O}\left(\dfrac{\textnormal{var}(P_{I} + \sigma_{k}^{2})}{\mathbb{E}\{P_{I} + \sigma_{k}^{2}\}^{3}} \right).
\end{equation}

In this paper, we take the simplest precoding algorithm \cite{T1}, i.e., MRT, as an example. Recall \eqref{eq:HdC}, when the MRT is used at the BS, the precoding matrix $ \mathbf{W} $ can be given by
\begin{equation} \label{eq:Wmrt}
\mathbf{W}_{\textnormal{mrt}} = \hat{\mathbf{H}}_{d,CL}^{H} \: .
\end{equation}
Let $ \lambda_{\textnormal{mrt}} $ represent the normalisation parameter of the MRT precoding scheme, to meet the power constraint at the BS. With no calibration (NC), i.e., $\mathbf{B} = \mathbf{I}_{M}$, the $k^{\text{th}}$ UT's output SINR (signal-to-interference-plus-noise ratio) with uncalibrated MRT precoder has been derived in \cite{ReciprocityErrorTCOM2016}. The analytical result of the output SINR in \cite{ReciprocityErrorTCOM2016} is obtained based on \eqref{eq:EPi}, thus, it can be used to obtain the ergodic sum rate in \eqref{eq:RK2}, as follows:
\begin{lemma} \label{l1}
	Consider a TDD massive MIMO system modelled in Section~\ref{sec:sys}, with uncalibrated MRT precoder at the BS. The closed-form expression of the ergodic sum rate for $K$ UTs, $R_{K,\textnormal{mrt}}^{\textnormal{NC}}$ , is given by
	\begin{align} \label{eq:RKMRTNC}
	R_{K,\textnormal{mrt}}^{\textnormal{NC}} \approx & K\textnormal{log}_2 \left( \! 1 + \rho_{d}A_t\left(\dfrac{a^2A_r((M-1)A_I+2)\!+b^2}{a^2A_r+b^2} \right) \right. \notag \\
	&  \left. \times \left( \dfrac{ K^{2} + \rho_{d}K(K-1)(\rho_{d}A_t^2 + 2A_t)}{(\rho_{d}(K-1)A_t + K)^3} \right) \right) ,
	\end{align}
	where the estimation-error-related parameters $a$ and $b$ are given by \eqref{eq:esterr}, and the reciprocity-error-related parameters $A_t$, $A_r$ and $A_I$ are given in \cite{ReciprocityErrorTCOM2016}.
\end{lemma}
\begin{proof}
	See the proof of Theorem 1 in \cite{ReciprocityErrorTCOM2016}.
\end{proof}

The result in \eqref{eq:RKMRTNC} quantifies the compound effect of the reciprocity and estimation errors on the ergodic achievable sum rate of the MRT precoded system without calibration. 


When the inverse calibration is applied at the BS for the MRT precoder, the precoding matrix $\mathbf{W}_{\text{mrt}}$ can be expressed as
	\begin{align} \label{eq:WmrtIC}
	\mathbf{W}_{\textnormal{mrt}}^{\textnormal{IC}}  & = \left( a\mathbf{H}_{d}\mathbf{E}\mathbf{B}_{\text{IC}} + b\mathbf{N}_{u}^{T}\mathbf{B}_{\text{IC}}\right)^{H} \notag \\
	& = a\mathbf{H}_{bt}^{-1}\mathbf{H}^{\ast} + b \mathbf{H}_{bt}^{-1}\left(\mathbf{H}_{br}^{\ast}\right)^{-1}\mathbf{N}_{u}^{\ast} ,
	\end{align}
	Then the corresponding normalisation parameter $ \lambda_{\textnormal{mrt}}^{\textnormal{IC}} $, the expected values of the desired signal power and the interference power of the $k^{\text{th}}$ UT, i.e., ${\mathbb E}\left\lbrace P_{s,\text{mrt}}^{\text{IC}}\right\rbrace$ and ${\mathbb E}\left\lbrace1/\left( P_{I,\text{mrt}}^{\text{IC}} + \sigma_{k}^{2}\right) \right\rbrace$, can be derived. Details are provided in Appendix~\ref{sec:AppB}. These values are used to calculate the ergodic sum rate in the following proposition:
    \begin{proposition} \label{pRIC}
    Assuming that the same conditions are held as in the Lemma~\ref{l1}, while the inverse calibration is applied at the BS. The ergodic sum rate for $K$ UTs, $R_{K,\textnormal{mrt}}^{\textnormal{IC}}$ , is given by
    \begin{align} \label{eq:RKMRTIC}         R_{K,\textnormal{mrt}}^{\textnormal{IC}} \approx & K\textnormal{log}_2 \left( 1 + \rho_{d}\left(\dfrac{a^2(M-1)+b^2E_{\bar{2}}^r}{a^2+b^2E_{\bar{2}}^r} \right) \right. \notag \\
    & \!\!\! \left. \times \left( \dfrac{ (KE_{\bar{2}}^t)^2 + \rho_{d}K(K-1)(\rho_{d} + 2E_{\bar{2}}^t)}{(\rho_{d}(K-1) + KE_{\bar{2}}^t)^3} \right) \right) ,
    \end{align}
    where the reciprocity-error-related parameters $E_{\bar{2}}^r$ and $E_{\bar{2}}^t$ are given by \eqref{eq:E-2r} and \eqref{eq:E-2t} respectively, in Appendix~\ref{sec:AppA}.
    \end{proposition}
    \begin{proof}
    	See Appendix~\ref{sec:AppB}.
    \end{proof}
When the traditional relative calibration is applied at the BS, the precoding matrix $\mathbf{W}_{\text{mrt}}$ becomes
	\begin{align} \label{eq:WmrtRC}
	\mathbf{W}_{\textnormal{mrt}}^{\textnormal{RC}}  & = \left( a\mathbf{H}_{d}\mathbf{E}\mathbf{B}_{\text{RC}} + b\mathbf{N}_{u}^{T}\mathbf{B}_{\text{RC}}\right)^{H} \notag \\
	& = a\mathbf{H}_{bt}^{\ast}\mathbf{H}^{\ast} + b \mathbf{H}_{bt}^{\ast}\left( \mathbf{H}_{br}^{\ast}\right)^{-1} \mathbf{N}_{u}^{\ast}.
	\end{align}
	Similar to the case of the inverse calibration, $ \lambda_{\textnormal{mrt}}^{\textnormal{IC}} $, ${\mathbb E}\left\lbrace P_{s,\text{mrt}}^{\text{IC}}\right\rbrace$ and ${\mathbb E}\left\lbrace1/( P_{I,\text{mrt}}^{\text{IC}} + \sigma_{k}^{2}) \right\rbrace$ can be derived (see Appendix~\ref{sec:AppB}). Then the corresponding closed-form expression of the ergodic sum rate can be given as follows:
	\begin{proposition} \label{pRRC}
	Assuming that the same conditions are held as in the Lemma~\ref{l1}, while the relative calibration is applied at the BS. The ergodic sum rate for $K$ UTs, $R_{K,\textnormal{mrt}}^{\textnormal{RC}}$ , is given by
	\begin{align} \label{eq:RKMRTRC} \!\!\! R_{K,\textnormal{mrt}}^{\textnormal{RC}} \approx & K\textnormal{log}_2 \left( \! 1 \!+\! \rho_{d}\left(\dfrac{a^2((M\!-\!1)A_t^2\!+\!2E_4^t)\!+b^2E_4^tE_{\bar{2}}^r}{a^2+b^2E_{\bar{2}}^r} \right) \right. \notag \\
	& \!\!\! \left. \times \left( \dfrac{ K^{2}A_t^2 \!+\! \rho_{d}K(K\!-\!1)E_4^t(\rho_{d}E_4^t \!+\! 2A_t)}{(\rho_{d}(K-1)E_4^t + KA_t)^3} \right) \right) ,
	\end{align}
	where the reciprocity-error-related parameters $E_4^t$ is given by \eqref{eq:E4t} in Appendix~\ref{sec:AppA}.
	\end{proposition}
	\begin{proof}
		See Appendix~\ref{sec:AppB}.
	\end{proof}

Comparing \eqref{eq:RKMRTIC} with \eqref{eq:RKMRTRC}, it can be seen that the inverse calibration completely removes phase reciprocity error, and leaves the MRT precoder with the negligible amplitude-error-related parameters, i.e., $E_{\bar{2}}^t$ and $E_{\bar{2}}^r$. On the contrary, the relative calibration introduces the residual amplitude error into the MRT system, i.e., $A_t^2$ and $E_4^t$ in \eqref{eq:RKMRTRC} (based on the results in Appendix~\ref{sec:AppA}, we have $E_4^t > A_t^2 > A_t$). This in turn can cause the aforementioned estimation error amplification. Consequently, the compound effect of the residual amplitude reciprocity error and the ``amplified" channel estimation error may result in a significant performance loss in the system with the relative calibration. 


\section{Simulation Results} \label{sec:re}
In this section, we perform Monte-Carlo simulations to corroborate the analysis presented in Section~\ref{sec:performance}, and compare the performance of the calibration schemes discussed in the paper under different scenarios, in order to provide valuable insights into the practical system design. Unless otherwise specified, the statistical magnitudes of both amplitude and phase reciprocity errors are identical to that defined in \cite{ReciprocityErrorTCOM2016,R1100426}, with quadruple notations, such that $ (\alpha_{bt,0}, \sigma_{bt}^2, [a_{t}, b_{t}]) = (\alpha_{br,0}, \sigma_{br}^2, [a_{r}, b_{r}]) = (0\ \textnormal{dB}, 1, [-4\ \textnormal{dB}, 4\ \textnormal{dB}])$, and $(\theta_{bt,0}, \sigma_{\varphi_{t}}^2, [\theta_{t,1}, \theta_{t,2}]) = (\theta_{br,0}, \sigma_{\varphi_{r}}^2, [\theta_{r,1}, \theta_{r,2}]) = (0^{\circ}, 1, [-50^{\circ}\!, \ 50^{\circ}]) $, as considered in \cite{ReciprocityErrorTCOM2016}. We also consider a reference scenario, namely, ``Perfect Channel Reciprocity", for the case that $\sigma_{bt}^2 = \sigma_{bt}^2 = \sigma_{\varphi_{t}}^2 = \sigma_{\varphi_{t}}^2 = 0$. 
The orthogonal UL pilots are of length $\tau_u = K$, and the ergodic sum rate is measured in bits/s/Hz.

First, we investigate the ergodic sum rate of MRT with \textit{IC} and \textit{RC}, for different DL SNR regimes. As a benchmark, we consider perfect channel reciprocity and no calibration, \textit{NC}. We use $M = 100$, $K = 10$, and reciprocity-error-related parameters as mentioned before. Fig.~\ref{fig:fig3} illustrates that, for the MRT precoded system, IC nearly eliminates the effect of reciprocity error, and outperforms RC. Thus, we conclude that IC is more efficient than RC for MRT. More specifically, increasing the DL transmit power 10 times, e.g., $\rho_d$ from 0 $\text{dB}$ to 10 $\text{dB}$, the performance of IC increases by 36\% (i.e., 9 bits/s/Hz), while only 12\% improvement (i.e., 3 bits/s/Hz) for RC. This can also be approved analytically based on the comparison between \eqref{eq:RKMRTIC} and \eqref{eq:RKMRTRC}. As mentioned earlier, both IC and RC can remove the phase reciprocity error. However, RC suffers from the strong residual amplitude error, which can cause estimation error amplification that could outweigh the benefit of using RC. We further analyse this effect of estimation error amplification in RC in Fig.~\ref{fig:fig4} and~\ref{fig:fig5}.
\begin{figure}[!t]
	\centering
	\includegraphics[angle=0,width=0.45\textwidth]{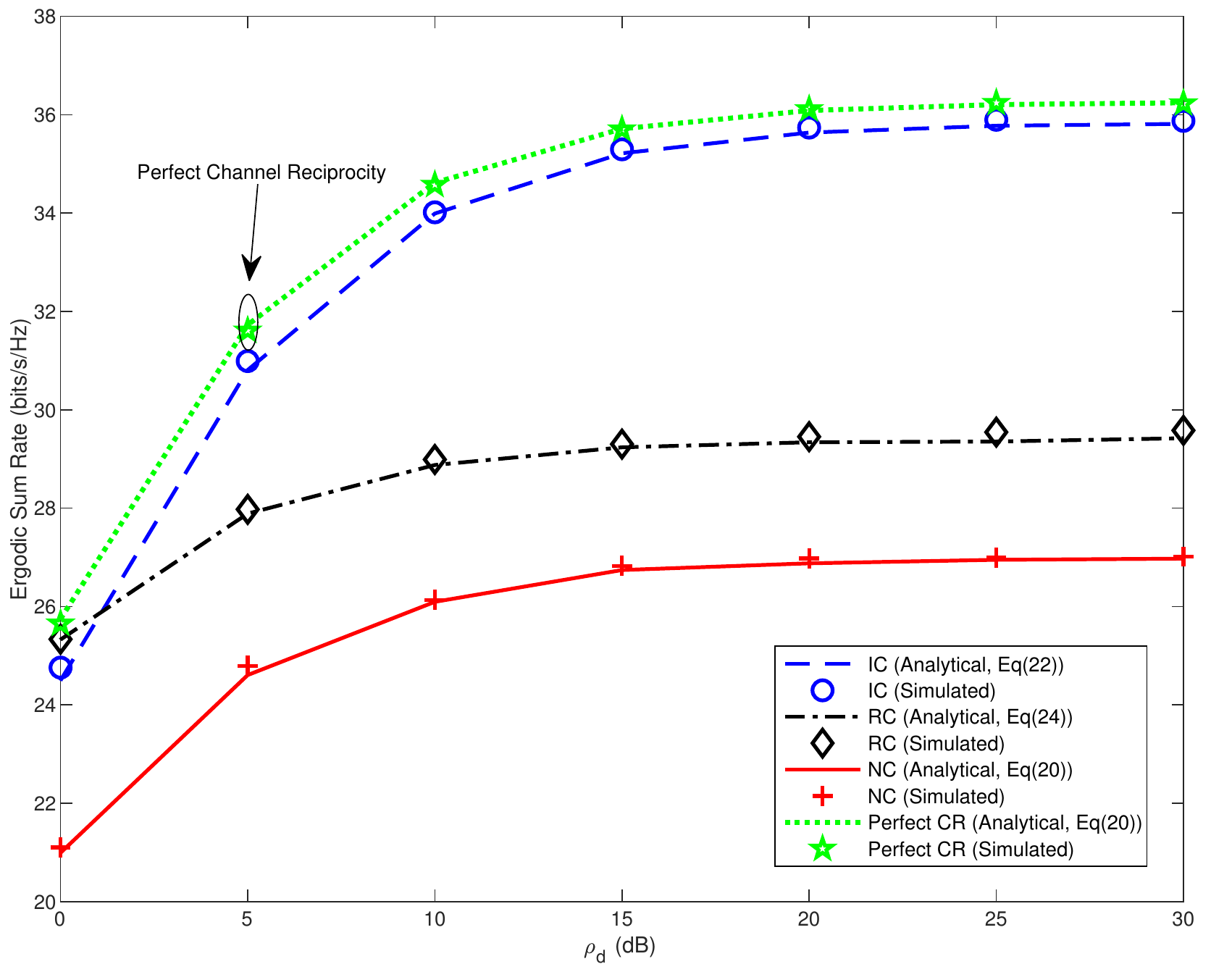} 
	\caption{Ergodic Sum Rate versus DL SNR in the presence of the high level reciprocity error and channel estimation error with $\rho_u = 0 \ \text{dB}$.}
	\label{fig:fig3}
\end{figure}

\begin{figure}[!t]
	\centering
	\includegraphics[angle=0,width=0.45\textwidth]{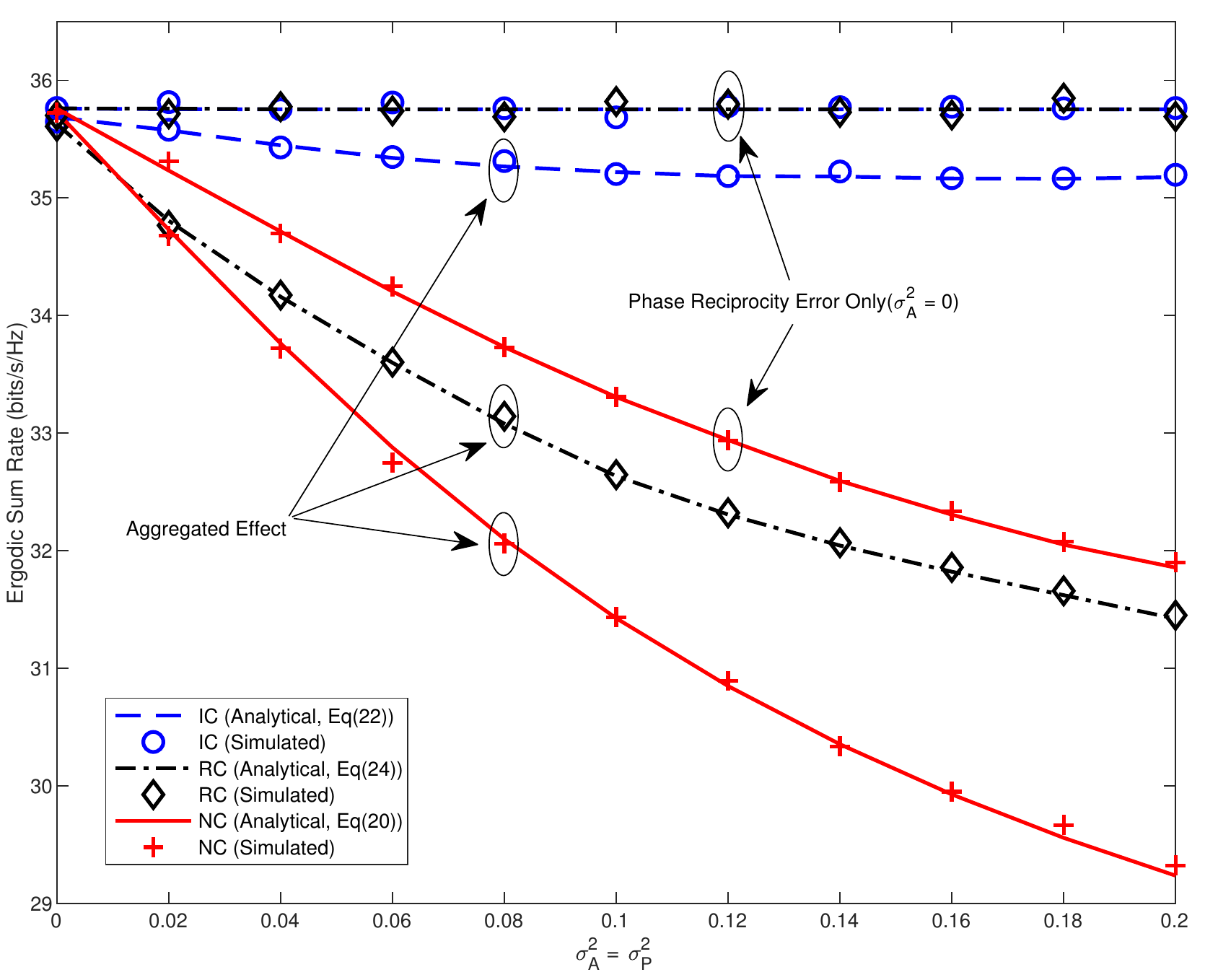} 
	\caption{Ergodic Sum Rate of MRT in the presence of different combinations of amplitude and phase reciprocity errors.}
	\label{fig:fig4}
\end{figure}
In Fig.~\ref{fig:fig4}, we assume $M = 100$, $K = 10$, $\rho_d = 10\ \text{dB}$, and a low level estimation error with $\rho_u = 10\ \text{dB}$. Let $\sigma_A^2 = \sigma_P^2$ be the x-axis that captures the aggregated effect of both amplitude and phase reciprocity errors, i.e., $ (\alpha_{bt,0}, \sigma_{bt}^2, [a_{t}, b_{t}]) = (\alpha_{br,0}, \sigma_{br}^2, [a_{r}, b_{r}]) = (0\ \textnormal{dB}, \sigma_A^2, [-4\ \textnormal{dB}, 4\ \textnormal{dB}])$, and the phase error has $(\theta_{bt,0}, \sigma_{\varphi_{t}}^2, [\theta_{t,1}, \theta_{t,2}]) = (\theta_{br,0}, \sigma_{\varphi_{r}}^2, [\theta_{r,1}, \theta_{r,2}]) = (0^{\circ}, \sigma_P^2, [-50^{\circ}\!, \ 50^{\circ}]) $. We consider the case with the phase error only by setting $\sigma_A^2 = 0$. Fig.~\ref{fig:fig4} shows that: \textit{a)} both IC and RC eliminate the phase error, with a 12.5 \% increase in the sum rate compared with NC (i.e., 4 bits/s/Hz more than 32 bits/s/Hz of NC at $\sigma_A^2 = 0$ and $\sigma_P^2 = 0.2$); \textit{b)} the performance gain of RC over NC decreases by 50 \% in the case with the aggregated effect of the amplitude and phase error (i.e., only 2 bits/s/Hz improvement compared to NC at $\sigma_A^2 = \sigma_P^2 = 0.2$, which is 50 \% of the previous 4 bits/s/Hz improvement in (a)), whereas the gain of IC over NC increases by 50 \% in this case (i.e., around 6 bits/s/Hz improvement compared to NC at $\sigma_A^2 = \sigma_P^2 = 0.2$); \textit{c)} RC may not work properly in the presence of the compound effect of both reciprocity and estimation errors, see $\sigma_A^2$, $\sigma_P^2 < 0.02$. We now take a closer look at the third observation, \textit{c}, as follows.

Similar parameters are considered as in Fig.~\ref{fig:fig4}. In addition, higher levels of channel estimation error are introduced, e.g., $\rho_u = 0\ \text{dB}$ in Fig.~\ref{fig:fig5}(a) and $\rho_u = -5\ \text{dB}$ in Fig.~\ref{fig:fig5}(b). It can be seen that the gain of the relative calibration vanishes in the case of the severe estimation error, whereas the inverse calibration, due to its greater robustness to the compound effect of the reciprocity error and estimation error, still works effectively with only a minor performance degradation.
\begin{figure}[!t]
	\centering
	\includegraphics[angle=0,width=0.45\textwidth]{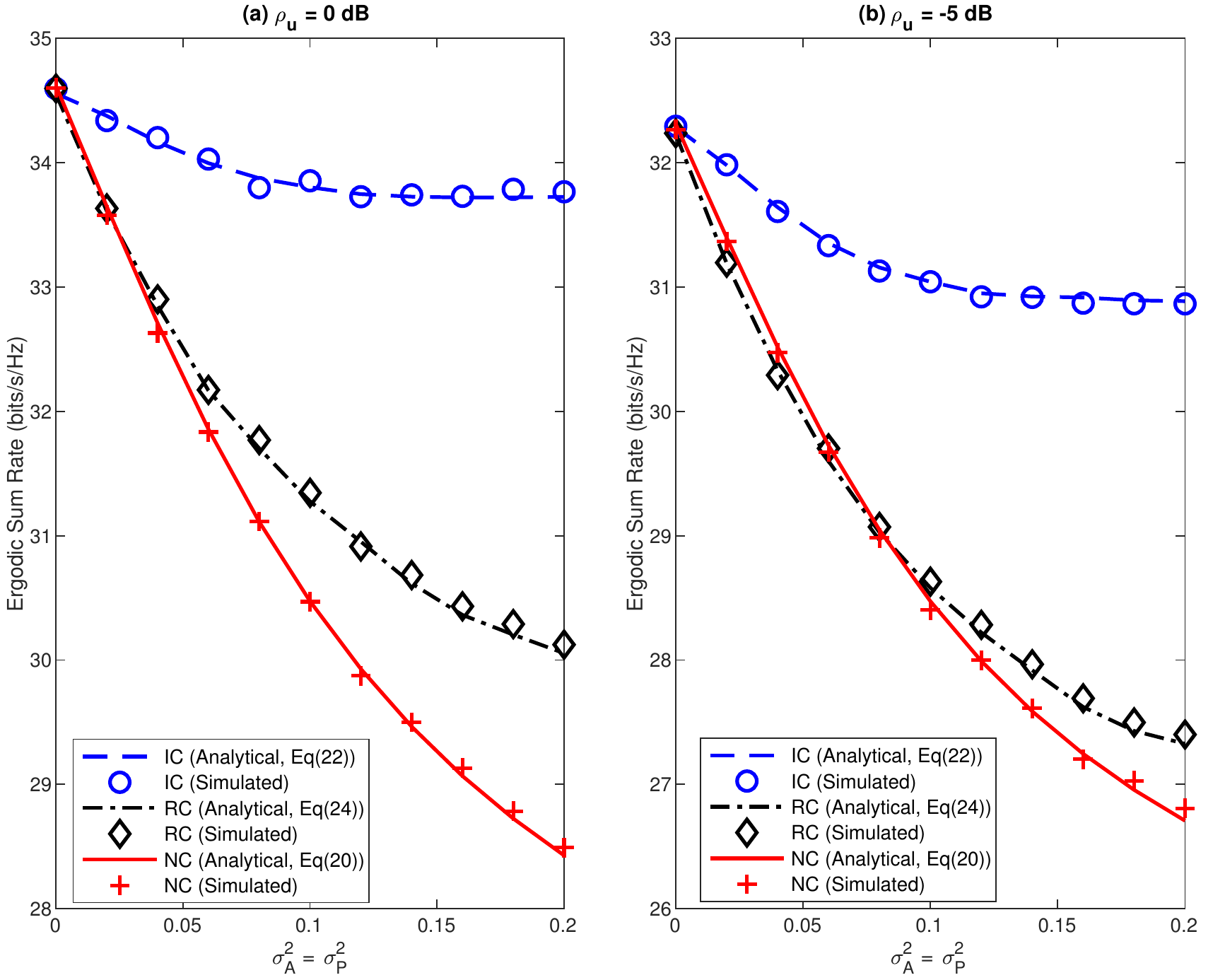} 
	\caption{Ergodic Sum Rate versus Reciprocity Error Variance with the different level of channel estimation error.}
	\label{fig:fig5}
\end{figure}


\section{Conclusion} \label{sec:con}
In this paper, we have presented and investigated two calibration schemes, i.e., inverse calibration and relative calibration, in the TDD massive MU-MIMO system. A low-cost calibration circuit has been considered, which enables the BS to select these two calibration algorithms. The performance of both calibration algorithms has been evaluated, by taking into account the compound effect of the multiplicative reciprocity error and the additive estimation error. Particularly, we have derived closed-formed expressions for the ergodic sum rate of the considered system, assuming MRT with the compound effect of both errors. We have demonstrated that the inverse calibration in general outperforms the traditional relative calibration. Analytical results perfectly match the simulated results for different scenarios, including large or practical number of BS antennas, different number of UTs and different combinations of the reciprocity error and estimation error. The comprehensive performance analysis has been given to provide a useful guidance for the selection of the calibration schemes in the massive MIMO system, which inevitable in practice.

For the way forward, one can consider the application of the inverse calibration in different scenarios, e.g., with different precoding schemes. It can also be of interest to carry out a field trial of the presented calibration circuits and algorithms, thanks to the simplicity of their design principles.

\appendices
\section{} \label{sec:AppA}
Recall the probability density function of a truncated Gaussian distributed variable $ x \sim \mathcal{N}_{\text T}(\mu, \sigma^2), x \in [a, b] $ (also $\alpha$, $\beta$ and $Z$) in \cite{ReciprocityErrorTCOM2016}, the $l^{\text{th}}$ ($l \geq 0$) non-central moment of $x$, denoted by $E_l^x$, is given by \cite{horrace2015moments}
\begin{equation}
E_l^x = \mathbb{E}\left\lbrace x^l \right\rbrace = \sum_{i = 0}^{l}\binom{l}{i}\sigma^i\mu^{l-i}L_{i}.
\end{equation}
We can now calculate the $l^{\text{th}}$ ($l \geq 0$) non-central moment of a truncated Gaussian distributed variable, such as 
\begin{align} \label{eq:E4t}
\!\!\!\! E_4^t & = \mathbb{E}\left\lbrace A_{bt,i}^4 \right\rbrace \notag \\
& = \mu^4 \!+\! 4\sigma\mu^3\!\left(\! \frac{\phi(\beta)\! -\! \phi(\alpha)}{Z}\! \right) \!+\! 6\sigma^2\mu^2\!\left(\!1\!+\! \frac{\alpha\phi(\alpha)\! -\!\beta\phi(\beta)}{Z}\!\right) \notag \\
& \qquad + 4\sigma^3\mu\left(\frac{(\alpha^2+2)\phi(\alpha) -(\beta^2+2)\phi(\beta)}{Z}\right) \notag \\
& \qquad + \sigma^4\left(3+\frac{(\alpha^3+3\alpha)\phi(\alpha) \!- \!(\beta^3+3\beta)\phi(\beta)}{Z}\right).
\end{align}
To calculate the expected value of the inverse square of the  truncated Gaussian distributed variable $x$, denoted by $E_{\bar{2}}^x$ where the subscript $(\cdot)_{\bar{2}}$ is used for the 2-nd inverse moment of $x$, we have
\begin{align} \label{eq:E-2}
E_{\bar{2}}^x & = \mathbb{E}\left\lbrace \frac{1}{x^2} \right\rbrace = \int_{a}^{b} x^{-2}\dfrac{1}{\sigma Z}\phi\left( \frac{x-\mu}{\sigma}\right)\, \mathrm{d}x \notag \\
& = \dfrac{1}{\sqrt{2\pi}\sigma Z} \int_{a}^{b} x^{-2}\textnormal{exp}\left( -\frac{1}{2}\left( \frac{x-\mu}{\sigma}\right)^{2}\right)\, \mathrm{d}x.
\end{align} 
The integral in \eqref{eq:E-2} is a nonelementary antiderivative and not able to be further simplified when $\mu \neq 0$. 
Consider the particular case in this work, 
the value of $E_{\bar{2}}^r$ can be measured based on the measurement of the response of the BS RF frontends, i.e., $\mathbf{R}_{\textnormal{self}}$ and $\mathbf{R}_{\textnormal{half}}$, such that
\begin{equation} \label{eq:E-2r}
E_{\bar{2}}^r = \frac{1}{M}\textnormal{tr}\left(\mathbf{R}_{\textnormal{half}}^{-1}(\mathbf{R}_{\textnormal{half}}^{\ast})^{-1} \right).
\end{equation}
Similarly, we have
\begin{align}
\label{eq:E-2t} E_{\bar{2}}^t = \frac{1}{M}\textnormal{tr}\left(\mathbf{R}_{\textnormal{self}}^{2}(\mathbf{R}_{\textnormal{half}}\mathbf{R}_{\textnormal{half}}^{\ast})^{-1} \right).
\end{align}

\section{} \label{sec:AppB}
\subsection{Inverse Calibration}
\begin{enumerate}[label=(\arabic*), wide=0pt]
\item {$\lambda_{\textnormal{mrt}}^{\textnormal{IC}}$:} Recall \eqref{eq:lamda} and \eqref{eq:WmrtIC}, the denominator inside of the square root in $\lambda_{\textnormal{mrt}}^{\textnormal{RC}}$ can be given by
\begin{align} 
\mathbb{E}\left\lbrace \textnormal{tr}\left( \mathbf{W}_{\textnormal{mrt}}^{\textnormal{IC}} (\mathbf{W}_{\textnormal{mrt}}^{\textnormal{IC}})^{H} \right)  \right\rbrace & \overset{\textnormal{(a)}}{=} {\mathbb E} \left\lbrace  \textnormal{tr}\left(a^2\mathbf{H}_{bt}^{-1}\mathbf{H}^{\ast}\mathbf{H}^{T}(\mathbf{H}_{bt}^{\ast})^{-1}\right) \right\rbrace \notag \\
& \!\!\!\!\!\!\!\!\!\!\!\!\!\!\!\!\!\!\!\!\!\!\!\!\!\!\!\!\!\!\!\!\!\!\!\!\!\!\!\!\!\! + {\mathbb E} \left\lbrace  \textnormal{tr}\left(b^2(\mathbf{H}_{bt}\mathbf{H}_{br}^{\ast})^{-1}\mathbf{N}_{u}^{\ast}\mathbf{N}_{u}^{T}(\mathbf{H}_{bt}^{\ast}\mathbf{H}_{br})^{-1}\right) \right\rbrace \\
& \!\!\!\!\!\!\!\!\!\!\!\!\!\!\!\!\!\!\!\!\!\!\!\! = MKE_{\bar{2}}^t\left( a^2+b^2E_{\bar{2}}^r \right),
\end{align}
where $\textnormal{(a)}$ is conditioned on the independence between $\mathbf{H}$ and $\mathbf{N}_{u}$. Thus we have
\begin{equation} \label{eq:lamdamrtIC}	\lambda_{\textnormal{mrt}}^{\textnormal{IC}} = \sqrt{\dfrac{1}{MKE_{\bar{2}}^t\left( a^2+b^2E_{\bar{2}}^r \right)}}.
\end{equation}

\item {${\mathbb E}\left\lbrace P_{s,\text{mrt}}^{\text{IC}}\right\rbrace$:} Recall \eqref{eq:Ps} and $\mathbb{E}\left\lbrace \lvert s_k \rvert^2 \right\rbrace = 1$, we have
\begin{equation} \label{eq:EPsmrtIC1}
{\mathbb E}\left\lbrace P_{s,\text{mrt}}^{\text{IC}}\right\rbrace = {\mathbb E} \left\lbrace \lvert\sqrt{\rho_{d}}\lambda_{\textnormal{mrt}}^{\textnormal{IC}}\mathbf{h}_{k}^{T}\mathbf{H}_{bt}\mathbf{w}_{k,\textnormal{mrt}}^{\textnormal{IC}}\rvert^{2} \right\rbrace,
\end{equation}
where
\begin{align}
{\mathbb E} \left\lbrace \lvert\mathbf{h}_{k}^{T}\mathbf{H}_{bt}\mathbf{w}_{k,\textnormal{mrt}}^{\textnormal{IC}}\rvert^{2} \right\rbrace & \notag \\
& \label{eq:EhwmrtIC1}\!\!\!\!\!\!\!\!\!\!\!\!\!\!\!\!\!\!\!\!\!\!\!\!\!\!\!\!\!\!\!\!\!\!\!\!\!\!\!= {\mathbb E} \left\lbrace \lvert\mathbf{h}_{k}^{T}\mathbf{H}_{bt}(a\mathbf{H}_{bt}^{-1}\mathbf{h}_{k}^{\ast}+b(\mathbf{H}_{bt}\mathbf{H}_{br}^{\ast})^{-1}\mathbf{n}_{u,k}^{\ast})\rvert^{2} \right\rbrace \\
& \label{eq:EhwmrtIC2}\!\!\!\!\!\!\!\!\!\!\!\!\!\!\!\!\!\!\!\!\!\!\!\!\!\!\!\!\!\!\!\!\!\!\!\!\!\!\!= a^2{\mathbb E}\left\lbrace \lvert\mathbf{h}_{k}^{T}\mathbf{h}_{k}^{\ast}\rvert^{2}\right\rbrace + b^2{\mathbb E}\left\lbrace\lvert\mathbf{h}_{k}^{T}(\mathbf{H}_{br}^{\ast})^{-1}\mathbf{n}_{u,k}^{\ast}\rvert^{2}\right\rbrace \\
& \label{eq:EhwmrtIC3}\!\!\!\!\!\!\!\!\!\!\!\!\!\!\!\!\!\!\!\!\!\!\!\!\!\!\!\!\!\!\!\!\!\!\!\!\!\!\!= a^2M(M+1) + b^2ME_{\bar{2}}^r.
\end{align}
Substituting \eqref{eq:lamdamrtIC} and \eqref{eq:EhwmrtIC3} into \eqref{eq:EPsmrtIC1}, we have
\begin{equation} \label{eq:EPsmrtIC}
{\mathbb E}\left\lbrace P_{s,\text{mrt}}^{\text{IC}}\right\rbrace = \dfrac{\rho_d\left( a^2(M+1)+b^2E_{\bar{2}}^r\right) }{KE_{\bar{2}}^t\left( a^2+b^2E_{\bar{2}}^r \right)}.
\end{equation}

\item {${\mathbb E}\left\lbrace1/\left( P_{I,\text{mrt}}^{\text{IC}} + \sigma_{k}^{2}\right) \right\rbrace$:} Consider that the symbols of different users are independent, we first calculate ${\mathbb E}\left\lbrace P_{I,\text{mrt}}^{\text{IC}}\right\rbrace$ as
\begin{align}
{\mathbb E}\left\lbrace P_{I,\text{mrt}}^{\text{IC}}\right\rbrace & = {\mathbb E} \left\lbrace \left|\sum_{i = 1, i \neq k}^{K}\sqrt{\rho_{d}}\lambda_{\textnormal{mrt}}^{\textnormal{IC}}\mathbf{h}_{k}^{T}\mathbf{H}_{bt}\mathbf{w}_{i,\textnormal{mrt}}^{\textnormal{IC}}\right|^{2} \right\rbrace \\
& = \sum_{i = 1, i \neq k}^{K} {\mathbb E} \left\lbrace \left|\sqrt{\rho_{d}}\lambda_{\textnormal{mrt}}^{\textnormal{IC}}\mathbf{h}_{k}^{T}\mathbf{H}_{bt}\mathbf{w}_{i,\textnormal{mrt}}^{\textnormal{IC}}\right|^{2} \right\rbrace \\
& \label{eq:EPImrtIC}\overset{\textnormal{(b)}}{=} \frac{\rho_d(K-1)}{KE_{\bar{2}}^t},
\end{align}
where $\textnormal{(b)}$ can be carried out in the similar way as that from \eqref{eq:EPsmrtIC1} to \eqref{eq:EPsmrtIC}. We then calculate $\textnormal{var}(P_{I,\text{mrt}}^{\text{IC}} + \sigma_{k}^{2})$ by following the technique as in \cite[Theorem~1]{ReciprocityErrorTCOM2016}, such that
\begin{equation} \label{eq:varmrtIC}
\textnormal{var}\left(P_{I,\text{mrt}}^{\text{IC}} + \sigma_{k}^{2}\right) = \frac{\rho_d^2(K-1)}{\left( KE_{\bar{2}}^t\right)^2 }.
\end{equation}
Substituting \eqref{eq:EPImrtIC} and \eqref{eq:varmrtIC} into \eqref{eq:EPi}, we have
\begin{equation}\label{eq:EPINmrtIC}
{\mathbb E}\left\lbrace\frac{1}{ P_{I,\text{mrt}}^{\text{IC}}\! +\! \sigma_{k}^{2}}\right\rbrace \!= \!\dfrac{KE_{\bar{2}}^t\left( (KE_{\bar{2}}^t)^2 \!+\! \rho_{d}K(K\!-\!1)(\rho_{d} \!+\! 2E_{\bar{2}}^t)\right) }{(\rho_{d}(K-1) + KE_{\bar{2}}^t)^3} .
\end{equation}

Now we can arrive at \eqref{eq:RKMRTIC} in Proposition~\ref{pRIC} by substituting \eqref{eq:EPsmrtIC} and \eqref{eq:EPINmrtIC} in \eqref{eq:RK2}.
\end{enumerate}

\subsection{Relative Calibration}
For the sake of simplicity, we list the main results for the relative calibration scheme as below. The derivation of these results follows the similar technique as in the previous subsection.
\begin{equation} \label{eq:lamdamrtRC}	\lambda_{\textnormal{mrt}}^{\textnormal{RC}} = \sqrt{\dfrac{1}{MKA_t\left(a^2+b^2E_{\bar{2}}^r\right)}},
\end{equation}
\begin{equation} \label{eq:EPsmrtRC}
{\mathbb E}\left\lbrace P_{s,\text{mrt}}^{\text{RC}}\right\rbrace = \dfrac{\rho_d\left( a^2((M\!-\!1)A_t^2+2E_4^t) \!+\!b^2E_{4}^tE_{\bar{2}}^r\right) }{KA_t\left( a^2+b^2E_{\bar{2}}^r \right)},
\end{equation}
\begin{align}
&\!\!\!\!\!\! {\mathbb E}\left\lbrace\frac{1}{ P_{I,\text{mrt}}^{\text{RC}} + \sigma_{k}^{2}}\right\rbrace \notag \\
&\label{eq:EPINmrtRC} =\dfrac{KA_t\left( K^2A_t^2 + \rho_{d}K(K-1)E_4^t(\rho_{d}E_4^t + 2A_t)\right) }{(\rho_{d}(K-1)E_4^t + KA_t)^3}.
\end{align}
Substituting \eqref{eq:EPsmrtRC} and \eqref{eq:EPINmrtRC} in \eqref{eq:RK2}, we have \eqref{eq:RKMRTRC} in Proposition~\ref{pRRC}.






%
\bibliographystyle{IEEEtran}
\bibliography{IEEEabrv,DMRef}

%








\end{document}